\documentclass[letterpaper]{article} 
\usepackage{aaai24}  
\usepackage{times}  
\usepackage{helvet}  
\usepackage{courier}  
\usepackage[hyphens]{url}  
\usepackage{graphicx} 
\usepackage{natbib}  
\usepackage{caption} 
\newcommand{\indep}{\mathop{\perp\!\!\!\!\perp}}
\usepackage{booktabs}
\usepackage{amsmath, amssymb, amsfonts, amsthm}
\usepackage{bm}
\usepackage{graphicx}
\newtheorem{thm}{Theorem}   
\usepackage{color}
\usepackage{multirow}
\usepackage{here}
\usepackage{algorithm}
\usepackage{algpseudocode}

\usepackage{newfloat}
\usepackage{listings}
\DeclareCaptionStyle{ruled}{labelfont=normalfont,labelsep=colon,strut=off} 
\floatstyle{ruled}
\newfloat{listing}{tb}{lst}{}
\floatname{listing}{Listing}
\title{Uncertainty Quantification in Heterogeneous Treatment Effect Estimation \\
with Gaussian-Process-Based Partially Linear Model}
\author{
    Shunsuke Horii\textsuperscript{\rm 1}, Yoichi Chikahara\textsuperscript{\rm 2}
}
\affiliations{
    \textsuperscript{\rm 1}Center for Data Science, Waseda University, Tokyo, Japan\\
    \textsuperscript{\rm 2}NTT Communication Science Laboratories, Kyoto, Japan

    s.horii@waseda.jp, chikahara.yoichi@gmail.com
}

\begin{document}

\maketitle

\begin{abstract}
Estimating heterogeneous treatment effects across individuals has attracted growing attention 
as a statistical tool for performing critical decision-making. We propose a Bayesian inference framework that quantifies the uncertainty in treatment effect estimation to support decision-making in a relatively small sample size setting. Our proposed model places Gaussian process priors on the nonparametric components of a semiparametric model called a partially linear model. This model formulation has three advantages. First, we can analytically compute the posterior distribution of a treatment effect without relying on the computationally demanding posterior approximation. Second, we can guarantee that the posterior distribution concentrates around the true one as the sample size goes to infinity. Third, we can incorporate prior knowledge about a treatment effect into the prior distribution, improving the estimation efficiency. Our experimental results show that even in the small sample size setting, our method can accurately estimate the heterogeneous treatment effects and effectively quantify its estimation uncertainty.
\end{abstract}

\section{Introduction}

\label{sec:intro}

Assessing heterogeneous treatment effects across individuals provides a key foundation for making critical decisions. 
For instance, understanding how greatly medical treatment effects are different across patients is helpful for precision medicine, and evaluating the impact of education programs on learning outcomes is essential for personalized learning. 

A widely used treatment effect measure is a conditional average treatment effect (CATE), which is an average treatment effect across individuals with identical feature attributes. 
CATE estimation is challenging when the number of features of an individual is large. 
Many methods aim to express the complex nonlinearity between treatment effects and features, 
using a nonparametric regression model, such as tree-based models \citep{hahn2020bayesian,hill2011bayesian,wager2018estimation} and neural networks \citep{johansson2016learning,shalit2017estimating,yoon2018ganite,wu2023stable}. 

However, most of these methods only output a single-point estimate 
and cannot consider the uncertainty in CATE estimation. 
This drawback is fatal
because decision-making under uncertainty is usual in many applications,
especially when we can only access a small amount of observational data.
An obvious example would be medical treatment planning.
For example, in the US, more than half of hospitals have fewer beds than 100 \citep{wiens2014study}, 
illustrating the difficulty of obtaining large-scale data
and the importance of considering the uncertainty.

To support decision-making in such critical applications, 
we propose a Bayesian framework for quantifying the CATE estimation uncertainty. 
To deal with a small sample size setting, 
we focus on a semi-parametric model called a partially linear model \citep{engle1986semiparametric}, 
which is linear with respect to the treatment 
but is nonlinear with respect to features, thanks to the nonparametric components. 

Our key idea is to place Gaussian process priors on the nonparametric components in a partially linear model.
This idea has three advantages.
First, we can analytically compute the posterior distribution of CATE.
Such analytical computation is 
much more computationally efficient than 
the approximate Bayesian inference,
which is required with the complex tree-based models \citep{hahn2020bayesian}.
Second, we can theoretically guarantee the asymptotic consistency of 
the posterior distribution of CATE under some mild conditions.
This theoretical guarantee also makes a striking contrast with the tree-based models, 
which have no consistency guarantee due to their model complexity.
Third, we can incorporate the prior knowledge about the CATE 
to improve the estimation accuracy.
To take prior knowledge example, 
consider the education program evaluation, 
where it is known that 
past academic performance is an important feature 
called a \textit{treatment effect modifier} \citep[Chapter 4]{hernan2020causal},
which affects the effect of an education program \citep{yeager2019national}.
We can utilize this prior knowledge
by designing the kernel functions used in the Gaussian process priors.
By contrast, integrating such prior knowledge is impossible 
with the existing Gaussian-process-based model \citep{alaa2018bayesian}, 
which puts the priors not on the CATE but on the outcomes.
Furthermore, our method can address not only binary treatment but also continuous-valued treatment \citep{hirano2004propensity}, which widens the scope of applications 
and is helpful, for instance, in determining the appropriate drug dosage for precision medicine 
\citep{bica2020estimating}. 

\begin{table}[h]
    \centering
    \caption{Comparison with existing Bayesian methods: AP, CG, and PK are acronyms for Analytic form of Posterior, Consistency Guarantee, and Prior Knowledge incorporation.}
    \label{table-methods}
    \small
    \begin{tabular}{lccc}
    \toprule 
    Method  & AP & CG & PK \\
    \midrule
    \citep{alaa2018bayesian}   & $\checkmark$ & $\checkmark$ &                  \\
    \citep{hahn2020bayesian} &             &             & $\checkmark$           \\
    \textbf{Proposed method} & $\checkmark$            & $\checkmark$            & $\checkmark$    \\
    \bottomrule 
    \end{tabular}
\end{table}

\textbf{Our contributions} are summarized as follows:
\begin{itemize}
    \item We establish a Bayesian framework that can effectively quantify the CATE estimation uncertainty in a relatively small sample size setting (Table \ref{table-methods}). To achieve this, we put Gaussian process priors on the nonparametric components of a partially linear model.
    \item We theoretically prove that the posterior distribution of CATE concentrates around the true one, as the sample size goes to infinity (Section \ref{sec:theory}).
    \item We experimentally show that the proposed method can accurately estimate the CATE and effectively quantify its estimation uncertainty, especially when given small and high-dimensional observational data.    
\end{itemize}
Our code is publicly available at \url{https://github.com/holyshun/GP-PLM}.

\section{Preliminaries}

\subsection{Problem Setup}

Our target estimand, CATE, is the average effect of treatment $T$ on outcome $Y$ in a subgroup of individuals with identical feature attributes $\bm{X}=\bm{x}$. Here we consider a binary or continuous-valued treatment ($T \in \{0, 1\}$ or $T\in \mathbb{R}$), a continuous-valued outcome ($Y \in\mathbb{R}$), and
a $d$-dimensional continuous-valued feature vector ($\bm{X} \in\mathbb{R}^d$).

In case of binary treatment $T \in \{0, 1\}$, 
a treatment effect for an individual is measured as 
the difference between random variables called \textit{potential outcomes}, 
$Y^{(1)} - Y^{(0)}$, where 
$Y^{(0)}$ and $Y^{(1)}$ represents outcome $Y$ 
when an individual is untreated ($T=0$) and treated ($T=1$), 
respectively \citep{rubin1974estimating}. 
Unfortunately, we can never observe treatment effect $Y^{(1)} - Y^{(0)}$ 
because we only observe outcome $Y = T Y^{(1)} + (1 - T) Y^{(0)}$ 
and can never jointly observe two potential outcomes $Y^{(0)}$ and $Y^{(1)}$.
For this reason, we focus on the average, CATE, which can be estimated from the data and is defined as the conditional expected value:
\begin{align}
    \mathrm{CATE}(\bm{x})&=\mathbb{E}[Y^{(1)}-Y^{(0)}|\bm{X}=\bm{x}]\nonumber \\
    &=\mathbb{E}[Y^{(1)}|\bm{X}=\bm{x}]-\mathbb{E}[Y^{(0)}|\bm{X}=\bm{x}]. \label{CATE-bi}
\end{align}

We can similarly define the CATE for continuous-valued treatment $T \in \mathbb{R}$, which represents the degree of treatment (e.g., the amount of chemotherapy). In this case, potential outcome $Y^{(t)}$ expresses the outcome when $T = t$ ($t \in \mathbb{R}$). The mean potential outcome across individuals with $\bm{X}=\bm{x}$, i.e., $\mathbb{E}[Y^{(t)}|\bm{X}=\bm{x}]$, can be regarded as a function of $t$, which is called a \textit{dose response function}. By taking the value difference of this function between treatment values $T=t$ and $T=t'$ ($t, t' \in \mathbb{R}$), we can measure the CATE as
\begin{align}
    \mathrm{CATE}(\bm{x}, t, t')=\mathbb{E}[Y^{(t')}|\bm{X}=\bm{x}]-\mathbb{E}[Y^{(t)}|\bm{X}=\bm{x}]. \label{CATE-cont}
\end{align}

Hence, in both treatment setups, we need to estimate mean potential outcome $\mathbb{E}[Y^{(t)}|\bm{X}=\bm{x}]$ for treatment $t \in \mathcal{T}$, where $\mathcal{T} = \{0, 1\}$ or $\mathcal{T} \subseteq \mathbb{R}$. To achieve this, we make two standard assumptions. One is the overlap condition (a.k.a., \textit{positivity}), i.e., $0< p(t|\bm{x})<1$ for all $t \in \mathcal{T}$ and for all $\bm{x}\in\mathcal{X}$ such that $p(\bm{x})>0$.
The other is the \textit{strongly ignorability} condition \citep{imbens2015causal}, $\{Y^{(t)} \colon t \in \mathcal{T}\}\indep T|\bm{X}$; this conditional independence relation is satisfied if features $\bm{X}$ include all confounders and contain only \textit{pretreatment variables}, which are not affected by treatment $T$ unlike mediators and colliders \citep{elwert2014endogenous}.\footnote{In addition, we assume that features $\bm{X}$ do not include pretreatment colliders, as with the standard CATE estimation methods.}
Under these two assumptions,
the mean potential outcome can be reformulated as
\footnote{For the derivation, see, e.g., the survey by \citet{yao2021survey}.}
\begin{align*}
    \mathbb{E}[Y^{(t)}|\bm{X}=\bm{x}]= \mathbb{E}[Y|\bm{X}=\bm{x}, T=t].
\end{align*}
That is, the mean potential outcome is reduced to 
the conditional expectation $\mathbb{E}[Y|\bm{X}=\bm{x}, T=t]$. 
To represent this conditional expectation, 
we employ a partially linear model, which is described below. 

\subsection{Partially Linear Model}
\label{subsec:PLM}

A partially linear model is a semi-parametric regression model introduced by \citet{engle1986semiparametric}. A widely used formulation of this model is
\begin{align}
    Y=\theta T + f(\bm{X})+\varepsilon, \label{PLM}
\end{align}
where $\theta \in \mathbb{R}$ is an unknown parameter representing a treatment effect, $f\colon \mathbb{R}^d \rightarrow \mathbb{R}$ is an unknown nonlinear function that expresses how greatly outcome $Y$ differ depending on the values of features $\bm{X}$, and $\varepsilon \overset{i.i.d.}{\sim} \mathcal{N}(0, s_{\varepsilon}^{-1})$ is a noise that follows a zero-mean Gaussian with precision $s_{\varepsilon} > 0$.

The model formulation \eqref{PLM} has two advantages. First, compared with nonparametric regression models such as tree-based models and neural networks, the estimation requires a much smaller sample size even when feature vector $\bm{X}$ is high-dimensional. Such a sample-size efficiency is essential to achieve practical applications with high data acquisition costs.
Second, despite this efficiency, the model \eqref{PLM} can represent the complex nonlinearity between outcome $Y$ and $\bm{X}$, using nonlinear function $f(\bm{X})$. 

By contrast, a drawback of the model \eqref{PLM} is that it cannot capture the treatment effect heterogeneity. This is because the treatment effect parameter $\theta$ is a constant with respect to features $\bm{X}$; hence, it cannot express how greatly the treatment effect changes depending on $\bm{X}$'s values. 

To overcome this drawback, we focus on the following variant of the partially linear model:
\begin{align}
    Y=\theta(\bm{X}) T+f(\bm{X})+\varepsilon, \label{proposed_model}
\end{align}
where $\theta\colon \mathbb{R}^d \rightarrow \mathbb{R}$ is an unknown nonlinear function. A significant difference from \eqref{PLM} is that $\theta(\cdot)$ in \eqref{proposed_model} is a function and can represent how strongly a treatment effect varies with $\bm{X}$'s values, thus addressing treatment effect heterogeneity. 

With the model \eqref{proposed_model},
the CATEs for binary and continuous-valued treatments in \eqref{CATE-bi} and \eqref{CATE-cont} are given by
\begin{align}
    \mathrm{CATE}(\bm{x}) = \theta(\bm{x}); \quad \mathrm{CATE}(\bm{x}, t, t') = (t' - t) \theta(\bm{x}).\label{eq-CATEs}
\end{align}
Hence, the CATE estimation reduces to 
the problem of estimating function $\theta(\cdot)$. 
In fact, other treatment effect measure called a \textit{conditional derivative effect}, $\mathrm{lim}_{\xi \rightarrow 0} \xi^{-1} \mathbb{E}[Y^{(t+\xi)} - Y^{(t)} | \bm{X}=\bm{x}]$ ($t \in \mathbb{R}$), can also be inferred 
by estimating $\theta(\cdot)$, which we detail in Appendix A.

Estimator $(t' - t) \theta(\bm{x})$ in \eqref{eq-CATEs} assumes that the CATE is linear with respect to treatment $T \in \mathbb{R}$; 
this assumption might be restrictive in some applications.
However, if we have prior knowledge about the functional relationship between outcome $Y$ and treatment $T$,
it is straightforward to use a pre-specified nonlinear function, 
$h\colon \mathbb{R} \rightarrow \mathbb{R}$ (e.g., $h(T) = \sqrt{T}$),
to reformulate $\theta(\bm{X}) T$ in \eqref{proposed_model} as $\theta(\bm{X}) h(T)$
and the CATE as $\mathrm{CATE}(\bm{x}, t, t') = (h(t') - h(t)) \theta(\bm{x})$.
One idea for how to formulate function $h$ is 
to follow the functional form between $Y$ and $T$
of the parametric models that are commonly used in the field. 
For instance, in medical treatment planning,
a sigmoid function is widely used in the dose-response curve models 
\citep{hill1910possible,hamilton1977trimmed}.
Developing a data-driven way to infer function $h$ is left as our future work.

Indeed, the partially linear model formulation \eqref{proposed_model} has also been studied 
in the treatment effect estimation framework
called \textit{Double/Debiased Machine Learning} (DML) \citep{chernozhukov2018double}. 
This framework is founded on the Frequentist approach 
and quantifies the uncertainty in estimating function $\theta$ 
with the confidence interval.
However, as shown by \citet{van2000asymptotic}, in a finite sample size setting, there is no theoretical guarantee about the uncertainty estimation with a confidence interval, and hence, the uncertainty estimation can be inaccurate.
To resolve this issue, we develop a Bayesian approach 
that infers the posterior distribution of $\theta$.

\section{Proposed Model}
\label{sec:posterior}

This section presents the derivation of the posterior distribution 
of function $\theta$ in the partially linear model in \eqref{proposed_model}.

Our posterior distribution can be formulated as follows. 
Suppose that each observation is obtained as 
$(t_i, \bm{x}_i, y_i) \overset{i.i.d.}{\sim} p(t, \bm{x}, y)$ for $i = 1, \dots, n$
and that we have $n$ observations 
$\mathcal{D}_n = (\bm{t}_{n}, \bm{X}_{n}, \bm{y}_{n})$,
where $\bm{t}_{n}=(t_{1},\ldots, t_{n})$, $\bm{X}_{n}=(\bm{x}_{1},\ldots, \bm{x}_{n})$, 
and $\bm{y}_{n}=(y_{1},\ldots, y_{n})$.
Our goal is to estimate the CATE values 
for a pre-specified set of $m$ feature vector values 
$\tilde{\bm{X}}_{m}=(\tilde{\bm{x}}_{1},\ldots, \tilde{\bm{x}}_{m})$.
Hence, our target posterior can be formulated as
posterior predictive distribution
$p(\tilde{\bm{\theta}}_{m}|\mathcal{D}_n, \tilde{\bm{X}}_{m})$,
where $\tilde{\bm{\theta}}_{m}=(\theta(\tilde{\bm{x}}_1), \dots, \theta(\tilde{\bm{x}}_m))$.
Note that this posterior predictive is different from that of
the standard Gaussian process regression:
it is the distribution of a function representing the CATE, not outcome $Y$.

\subsection{Priors}

To formulate the posterior predictive distribution, we place the Gaussian process prior distributions on functions $\theta$ and $f$ in the partially linear model \eqref{proposed_model} as
\begin{align}
    \theta(\cdot)&\sim \mathcal{GP}(0, C(\cdot,\cdot; \bm{\omega}_{\theta})), \label{GP-theta}\\
    f(\cdot)&\sim \mathcal{GP}(0, C(\cdot, \cdot; \bm{\omega}_{f})) \label{GP-f},
\end{align}
where $C(\cdot, \cdot; \bm{\omega})$ is the covariance function with parameter $\bm{\omega}$. Here, for notation simplicity,
we set the mean functions of the Gaussian processes to zero. 
Note that such zero-mean priors do not lose generality 
because they never restrict the posterior means, 
which are updated with the observed data (See, e.g., \citet{williams2006gaussian}).

A common formulation of covariance function $C(\cdot, \cdot; \bm{\omega})$ 
in \eqref{GP-theta} and \eqref{GP-f} is
the radial basis function (RBF) kernel:
\begin{align}
    C(\bm{x}, \bm{x}'; \omega )=\exp\left\{-\omega\|\bm{x}-\bm{x}'\|^{2}\right\} \quad (\bm{x}, \bm{x}' \in \mathbb{R}^d),\label{gaussian_kernel}
\end{align}
where $\omega > 0$ is a hyperparameter. 
We can also design the covariance function 
by utilizing our prior knowledge.
For example, if some features in $\bm{X}$ are known to be 
treatment effect modifiers, i.e.,
important features that
explain treatment effect heterogeneity,
we can formulate  $C(\cdot,\cdot; \bm{\omega}_{\theta})$ as
\begin{align}
    C(\bm{x}, \bm{x}'; \bm{\omega}_{\theta})=\exp\left\{-\sum_{k=1}^d
    \omega_{\theta,k} (x_k - x'_k)^{2}\right\},\label{modified_gaussian_kernel}
\end{align}
where $\bm{\omega}_{\theta} = (w_{\theta,1}, \dots, w_{\theta,d})$ is a vector of hyperparameters, 
whose $k$-th element $\omega_{\theta,k} \in \mathbb{R}^{\geq 0}$ 
represents the $k$-th feature's importance, 
which is given based on prior knowledge.
In Appendix E, we show that
using covariance function \eqref{modified_gaussian_kernel} 
leads to better estimation performance.

\subsection{Derivation of Posterior Predictive}

We show that we can analytically compute posterior predictive 
$p(\tilde{\bm{\theta}}_{m}|\mathcal{D}_n, \tilde{\bm{X}}_{m})$
if the prior hyperparameters, $\bm{\omega}_{\theta}$ and $\bm{\omega}_{f}$, and noise distribution parameter $s_{\varepsilon}$ are given.

To derive this analytic form, we take three steps. 
First, we derive the joint distribution 
$p(\Theta, \bm{y}_n|\bm{t}_n, \bm{X}_{n}, \tilde{\bm{X}}_{m})$, 
where $\Theta=(\bm{\theta}_{n}, \tilde{\bm{\theta}}_{m}, \bm{f}_{n})$
denotes a set of function values, including 
$\bm{\theta}_{n}=(\theta(\bm{x}_{1}),\ldots, \theta(\bm{x}_{n}))$ 
and $\bm{f}_{n}=(f(\bm{x}_{1}),\ldots,f(\bm{x}_{n}))$.
Second, by conditioning $\bm{y}_n$,
we obtain joint posterior 
$p(\Theta|\mathcal{D}_n,\tilde{\bm{X}}_{m})$.
Finally, we marginalize out $\bm{\theta}_{n}$ and $\bm{f}_{n}$ in $\Theta$ 
to derive posterior predictive 
$p(\tilde{\bm{\theta}}_{m}|\mathcal{D}_n,\tilde{\bm{X}}_{m})$.

\subsubsection{Joint Distribution}
 $p(\Theta, \bm{y}_{n}|\bm{t}_n, \bm{X}_{n}, \tilde{\bm{X}}_{m})$ is given as a product of the likelihood and the joint prior:
\begin{align}
    p(\Theta, \bm{y}_{n}|\bm{t}_{n}, \bm{X}_{n}, \tilde{\bm{X}}_{m}) = p(\bm{y}_{n}|\bm{\theta}_{n},\bm{f}_{n}, \bm{t}_{n})p(\Theta|\bm{X}_{n},\tilde{\bm{X}}_{m}).\label{joint}
\end{align}

Joint prior $p(\Theta|\bm{X}_{n},\tilde{\bm{X}}_{m})$ in \eqref{joint} is
given by the Gaussian process priors in \eqref{GP-theta} and \eqref{GP-f}
and hence is formulated as the multivariate Gaussian:
\begin{align}
    p(\Theta|\bm{X}_{n}, \tilde{\bm{X}}_{m})= \mathcal{N}(\bm{0}, \bm{\Sigma}_{\Theta\Theta}), \label{joint-Theta} 
\end{align}
where $\bm{\Sigma}_{\Theta\Theta}$ denotes the following covariance matrix:\footnote{$\bm{O}$ and $\bm{I}$ are the zero and identity matrices, respectively.
In this paper, all zero and identity matrices are denoted by $\bm{O}$ and $\bm{I}$, regardless of their matrix sizes.}
\begin{align*}
    \bm{\Sigma}_{\Theta\Theta} = \left(
    \begin{array}{lll}
    \bm{\Phi}_{nn} & \bm{\Phi}_{nm} & \bm{O}\\
    \bm{\Phi}_{nm}^{T} & \bm{\Phi}_{mm} & \bm{O}\\
    \bm{O} & \bm{O} & \bm{\Psi}_{nn}
    \end{array}
    \right),
\end{align*}
whose elements are given by
\begin{equation}
    \begin{split}
    \bm{\Phi}_{nn}&=\left(C(\bm{x}_{i},\bm{x}_{j};\bm{\omega}_{\theta})\right)_{1\le i,j\le n},\\ \bm{\Phi}_{nm}&=\left(C(\bm{x}_{i},\tilde{\bm{x}}_{j};\bm{\omega}_{\theta})\right)_{1\le i\le n, 1\le j\le m},\\ \bm{\Phi}_{mm}&=\left(C(\tilde{\bm{x}}_{i},\tilde{\bm{x}}_{j};\bm{\omega}_{\theta})\right)_{1\le i,j\le m}, \\\bm{\Psi}_{nn}&=\left(C(\bm{x}_{i},\bm{x}_{j};\bm{\omega}_{f})\right)_{1\le i,j\le n}.
    \end{split} \label{Phi}
\end{equation}

By contrast, likelihood 
$p(\bm{y}_{n}|\bm{\theta}_{n}, \bm{f}_{n}, \bm{t}_{n})$ in \eqref{joint}
is given by the partially linear model with a Gaussian noise in \eqref{proposed_model}. Hence, it is formulated as the multivariate Gaussian:
\begin{align}
    p(\bm{y}_{n}|\bm{\theta}_{n},\bm{f}_{n}, \bm{t}_{n})= \mathcal{N}(\bm{T}_{n}\bm{\theta}_{n}+\bm{f}_{n}, s_{\varepsilon}^{-1}\bm{I}), \label{likelihood}
\end{align}
where $\bm{T}_{n}=\mbox{diag}(\bm{t}_{n})$ is a diagonal matrix, whose diagonal component is $\bm{t}_{n}$.

Thus, both 
$p(\Theta|\bm{X}_{n},\tilde{\bm{X}}_{m})$ 
and $p(\bm{y}_{n}|\bm{\theta}_{n}, \bm{f}_{n}, \bm{t}_{n})$ 
are given as multivariate Gaussians. 
Therefore, joint distribution 
$p(\Theta, \bm{y}_{n}|\bm{t}_{n},\bm{X}_{n},\tilde{\bm{X}}_{m})$ 
is also a multivariate Gaussian. 
This Gaussian has mean $\bm{0}$
because both the joint prior in \eqref{joint-Theta} and the likelihood in \eqref{likelihood} are zero-mean;
the mean in \eqref{likelihood}, $\bm{T}_{n}\bm{\theta}_{n}+\bm{f}_{n}$, 
is zero because $p(\bm{\theta}_{n}, \bm{f}_{n}|\bm{X}_{n})$ is zero-mean.
As regards covariance matrix $\bm{\Sigma}$, 
we can explicitly express precision matrix $\bm{S} = \bm{\Sigma}^{-1}$ as
\begin{multline}
    \bm{S}=\left(
    \begin{array}{lll}
    \bm{\Phi}^{-1} & \bm{O} & \bm{O}\\
    \bm{O} & \bm{\Psi}_{nn}^{-1} & \bm{O}\\
    \bm{O} & \bm{O} & s_{\epsilon}\bm{I}
    \end{array}
    \right)+\\
    \left(
    \begin{array}{cccc}
    s_{\epsilon}\bm{T}_n^{2} & \bm{O} & s_{\epsilon}\bm{T}_n & -s_{\epsilon}\bm{T}_n\\
    \bm{O} & \bm{O} & \bm{O} & \bm{O}\\
    s_{\epsilon}\bm{T}_n & \bm{O} & s_{\epsilon}\bm{I} & -s_{\epsilon}\bm{I}\\
    -s_{\epsilon}\bm{T}_n & \bm{O} & -s_{\epsilon}\bm{I} & \bm{O}
    \end{array}
    \right),\label{S}
\end{multline}
where $\bm{\Phi}$ is the block matrix with the elements in \eqref{Phi}:
\begin{align*}
    \bm{\Phi} = \left(
    \begin{array}{cc}
    \bm{\Phi}_{nn} & \bm{\Phi}_{nm}\\
    \bm{\Phi}_{nm}^{T} & \bm{\Phi}_{mm} 
    \end{array}
    \right).
\end{align*}

\subsubsection{Joint Posterior}

$p(\Theta|\mathcal{D}_n,\tilde{\bm{X}}_{m})$ is obtained from joint distribution 
$p(\Theta, \bm{y}_{n}|\bm{t}_n, \bm{X}_{n}, \tilde{\bm{X}}_{m})$ in \eqref{joint}
by conditioning $\bm{y}_n$.

To confirm this, consider 
the covariance of the joint distribution 
in \eqref{joint},
i.e., $\bm{\Sigma}$,
whose precision matrix $\bm{S} = \bm{\Sigma}^{-1}$ is given by \eqref{S}.
Let us rephrase 
this covariance matrix as
\footnote{For example, $\bm{\Sigma}_{\Theta\Theta}$ is the submatrix of $\bm{S}^{-1}$, consisting of rows to $1$ to $2n$ and columns $1$ to $2n$.}
\begin{align*}
    \bm{\Sigma}=\left(
    \begin{array}{cc}
    \bm{\Sigma}_{\Theta\Theta} & \bm{\Sigma}_{\Theta\bm{y}_n}\\
    \bm{\Sigma}_{\bm{y}_n \Theta} & \bm{\Sigma}_{\bm{y}_n \bm{y}_n}
    \end{array}
    \right).
\end{align*}
Then, using the formula of conditional multivariate Gaussians 
(see e.g., \citet{bishop:2006:PRML, williams2006gaussian}), 
we can condition $\bm{y}_n$
and show that joint posterior $p(\Theta|\mathcal{D}_n,\tilde{\bm{X}}_{m})$
is the following multivariate Gaussian:
\begin{align}
    p(\Theta|\mathcal{D}_n,\tilde{\bm{X}}_{m})= 
    \mathcal{N}(\bm{\mu}_{\Theta|\bm{y}_n}, 
    \bm{\Sigma}_{\Theta|\bm{y}_n}). \label{joint-posterior}
\end{align}
whose mean $\bm{\mu}_{\Theta|\bm{y}_n}$ and 
covariance $\bm{\Sigma}_{\Theta|\bm{y}_n}$ are given by 
\begin{align*}   
    \bm{\mu}_{\Theta|\bm{y}_n}&=\bm{M}\bm{y}_{n},\\
    \bm{\Sigma}_{\Theta|\bm{y}_n}&=\bm{\Sigma}_{\Theta\Theta}-\bm{M} \bm{\Sigma}_{\bm{y}_n \Theta},
\end{align*}
where $\bm{M}=\bm{\Sigma}_{\Theta\bm{y}}\bm{\Sigma}_{\bm{yy}}^{-1}$.

\subsubsection{Posterior Predictive}

$p(\tilde{\bm{\theta}}_{m}|\mathcal{D}_n,\tilde{\bm{X}}_{m})$ is obtained
by marginalizing out $\bm{\theta}_{n}$ and $\bm{f}_{n}$ in $\Theta$
from joint posterior 
$p(\Theta|\mathcal{D}_n,\tilde{\bm{X}}_{m})$ in \eqref{joint-posterior}.

To see this, consider the submatrices 
in $\bm{M}$ and $\bm{\Sigma}_{\Theta|\bm{y}_n}$:
\begin{align*}
    \bm{M}&=\left(
    \begin{array}{c}
    \bm{M}_{\bm{\theta}}\\
    \bm{M}_{\tilde{\bm{\theta}}}\\
    \bm{M}_{\bm{y}}
    \end{array}
    \right),\\
    \bm{\Sigma}_{\Theta|\bm{y}_n}&=\left(
    \begin{array}{ccc}
    \bm{\Sigma}_{\bm{\theta\theta}|\bm{y}_n} & \bm{\Sigma}_{\bm{\theta}\tilde{\bm{\theta}}|\bm{y}_n} & \bm{\Sigma}_{\bm{\theta}\bm{f}|\bm{y}_n}\\
    \bm{\Sigma}_{\tilde{\bm{\theta}}\bm{\theta}|\bm{y}_n} & \bm{\Sigma}_{\tilde{\bm{\theta}}\tilde{\bm{\theta}}|\bm{y}_n} & \bm{\Sigma}_{\tilde{\bm{\theta}}\bm{f}|\bm{y}_n}\\
    \bm{\Sigma}_{\bm{f}\bm{\theta}|\bm{y}_n} & \bm{\Sigma}_{\bm{f}\tilde{\bm{\theta}}|\bm{y}_n} & \bm{\Sigma}_{\bm{f}\bm{f}|\bm{y}_n},
    \end{array}
    \right).
\end{align*}
With these notations, by marginalizing out $\bm{\theta}_{n}$ and $\bm{f}_{n}$,
we can formulate posterior predictive 
$p(\tilde{\bm{\theta}}_{m}|\mathcal{D}_n,\tilde{\bm{X}}_{m})$
as the following multivariate Gaussian:
\begin{align}
    p(\tilde{\bm{\theta}}_{m}|\mathcal{D}_n,\tilde{\bm{X}}_{m})
    = \mathcal{N}(\bm{M}_{\tilde{\bm{\theta}}}\bm{y}_{n}, \bm{\Sigma}_{\tilde{\bm{\theta}}\tilde{\bm{\theta}}|\bm{y}_n}). \label{posterior}
\end{align}

By marginalizing $\bm{f}_n$ and formulating the posterior in this way, 
we remove the estimation bias arising from nuisance parameter $f$,
which corresponds to the \textit{confounding bias} 
arising from the confounders in features $\bm{X}$.
With such Bayesian inference,
we have no need to estimate the \textit{propensity score} model,
unlike the Frequentist approach.
We detail the reason for this in Section \ref{subsec:related-NonBayesian}.

The posterior predictive in \eqref{posterior} 
requires computation time $O(n^{3})$ for sample size $n$, 
which might be problematic when $n$ is large.
However, as with the standard Gaussian process regression,
we can apply various approximation techniques, such as the sparse GP \citep{quinonero2005unifying},
to deal with a large-scale dataset.

\subsection{Addressing Unknown Hyperparameters}
\label{subsec:hyper}

So far, we have derived the posterior predictive under the assumption that 
the values of hyperparameters 
$\bm{\omega}_{\theta}$, $\bm{\omega}_{f}$, and $s_{\varepsilon}$ are given.
Since their true values are unknown in practice,
we must determine them using the observed data.
Below, we present the two data-driven approaches.

One is to put the priors on these hyperparameters.
In this case, the posterior predictive differs from \eqref{posterior},
and the analytic form is no longer available.
Hence, we will need to approximate the posterior predictive,
using Markov chain Monte Carlo (MCMC) methods,
such as the Metropolis-Hastings (MH) algorithm. 
Unfortunately, this approximation requires much computation time
and might be impractical.

For this reason, in our experiments, we took the other approach,
which estimates the values of hyperparameters 
$\bm{\omega}_{\theta}$, $\bm{\omega}_{f}$, and $s_{\varepsilon}$ 
by maximizing the marginal likelihood:
\begin{align*}
    &p(\bm{y}_n| \bm{t}_n, \bm{X}_n
    ;\bm{\omega}_{\theta}, \bm{\omega}_{f}, s_{\varepsilon})\\
    =&\int\int p(\bm{y}_n,\bm{\theta}_n,\bm{f}_n | \bm{t}_n, \bm{X}_n
    ;\bm{\omega}_{\theta}, \bm{\omega}_{f}, s_{\varepsilon})\mathrm{d}\bm{\theta}_n \mathrm{d}\bm{f}_n,
\end{align*}
where we marginalize out $\bm{\theta}_n$ and $\bm{f}_n$ 
from joint distribution:
\begin{align}
    &p(\bm{y}_n,\bm{\theta}_n,\bm{f}_n | \bm{t}_n, \bm{X}_n
    ;\bm{\omega}_{\theta}, \bm{\omega}_{f}, s_{\varepsilon}) \nonumber \\
    =&p(\bm{y}_n|\bm{\theta}_n,\bm{f}_n,\bm{t}_n ;s_{\varepsilon})
    p(\bm{\theta}_n|\bm{X}_n;\bm{\omega}_{\theta})
    p(\bm{f}_n|\bm{X}_n;\bm{\omega}_{f}).\label{marginal_likelihood}
\end{align}
Here the three distributions in the r.h.s. of \eqref{marginal_likelihood}
are given as 
\begin{align*}
&p(\bm{y}_n|\bm{\theta}_n,\bm{f}_n,\bm{t}_n ;s_{\varepsilon}) 
= \mathcal{N}(\bm{T}_n\bm{\theta}_n+\bm{f}_n,s_{\varepsilon}^{-1}\bm{I}),
\\
&p(\bm{\theta}_n|\bm{X}_n;\bm{\omega}_{\theta})=\mathcal{N}(\bm{0},\bm{\Phi}_{nn}),\\
&p(\bm{f}_n|\bm{X}_n;\bm{\omega}_{f})=\mathcal{N}(\bm{0},\bm{\Psi}_{nn}),
\end{align*}
respectively. Hence, the joint distribution in \eqref{marginal_likelihood} 
is also a multivariate Gaussian with mean $\bm{0}$ and covariance matrix
\begin{align*}
        \left(
    \begin{array}{ccc}
    s_{\epsilon}\bm{I} & -s_{\epsilon}\bm{T}_n & -s_{\epsilon}\bm{I}\\
    -s_{\epsilon}\bm{T}_n & s_{\epsilon}\bm{T}_n^{2}+\bm{\Phi}_{nn}^{-1} & s_{\epsilon}\bm{T}_n\\
    -s_{\epsilon}\bm{I} & s_{\epsilon}\bm{T}_n & s_{\epsilon}\bm{I}+\bm{\Psi}_{nn}^{-1}
    \end{array}
    \right)^{-1}.
\end{align*}
Using the matrix formula for the Schur complement,
we can marginalize  $\bm{\theta}_n$ and $\bm{f}_n$ from this multivariate Gaussian
and obtain the marginal likelihood 
$p(\bm{y}_n| \bm{t}_n, \bm{X}_n
;\bm{\omega}_{\theta}, \bm{\omega}_{f}, s_{\varepsilon})$
as a multivariate Gaussian with mean $\bm{0}$ and covariance matrix
\begin{align*}
s_{\epsilon}^{-1}\bm{I}+\bm{T}_n \bm{\Phi}_{nn}\bm{T}_n +\bm{\Psi}_{nn}.
\end{align*}
This marginal likelihood is not necessarily convex.
For this reason, in our experiments, we maximize it with respect to 
$\bm{\omega}_{\theta}$, $\bm{\omega}_{f}$, and $s_{\varepsilon}$ 
by combining the grid search and the gradient descent method.
See Appendix D.2 for the details.

\section{Theoretical Analysis} \label{sec:theory}

This section aims to guarantee the asymptotic convergence of our posterior distribution.
In particular, 
we prove the \textit{strong posterior consistency} 
\citep{ghosal1999posterior},
whose definition can be used
for non-parametric Bayesian models, including Gaussian processes.

This consistency notion is determined
whether the random measure representing the posterior
converges to the true data-generating distribution,
as sample size $n \rightarrow \infty$.
With the partially linear model,
the data-generating distribution is defined with 
functions $\theta$ and $f$.
Since we put the Gaussian process priors,
these functions themselves have randomness.
For this reason, we consider the posterior of their joint density,
$p_{\theta,f}(\bm{x}, t, y)=p(y|\bm{x},t,\theta,f)p(t|\bm{x})p(\bm{x})$.
Our goal is to show that as $n \rightarrow \infty$,
this posterior concentrates around its true joint density,
$p_{\theta_{0}, f_{0}}$ with high probability, 
where $\theta_{0}$ and $f_{0}$ denote the true data-generating functions.

To achieve this goal, we extend the results by \citet{ghosal1999posterior},
which proves the posterior consistency 
for the standard Gaussian process regression problems.
As with these results, 
we make Assumptions (\textbf{P}; Smoothness of priors), 
(\textbf{F}; Bounded feature space), 
(\textbf{T}; True functions), 
and (\textbf{E}; Exponential decay of priors).
Among them,
Assumption (\textbf{F}) requires 
feature vector values $\bm{x}$ to belong to a bounded subset of $\mathbb{R}^d$,
and Assumption (\textbf{T}) imposes the condition that
true functions $\theta_{0}$ and $f_{0}$ belong to the reproducing kernel Hilbert space (RKHS) of a kernel function used in the covariance function in Gaussian process priors.
Regarding the technical assumptions on priors 
(Assumptions (\textbf{P}) and (\textbf{E})), we detail them in Appendix B.

To extend the results by \citet{ghosal1999posterior} 
to the CATE estimation problem, 
we make an additional assumption on the boundedness of 
the conditional moments of treatment $T$ given features $\bm{X}$:
\begin{description}
    \item[(B)] (\textbf{Boundedness of conditional moments})
    There exist constants $C_{1}>0$ and $C_{2}>0$ such that
    \begin{align*}
        \mathbb{E}\left[T|\bm{X}\right]<C_{1}\quad \mathbb{P}_{\bm{X}}\textrm{-}a.s.;\quad
        \mathbb{E}\left[T^{2}|\bm{X}\right]<C_{2}\quad \mathbb{P}_{\bm{X}}\textrm{-}a.s.
    \end{align*}
\end{description}
This assumption imposes the conditional mean and variance of $T$ given $\bm{X}$ to be at most $C_1$ and $C_2$, respectively.

Using Assumptions (\textbf{P}), (\textbf{F}), (\textbf{T}), (\textbf{E}), and (\textbf{B}), we prove the consistency of the posterior of $\theta$ and $f$.
For this purpose, we use the $L_{1}$ metric between $p_{\theta, f}$ and $p_{\theta_{0},f_{0}}$:
\begin{align*}
    &\|p_{\theta, f}-p_{\theta_{0},f_{0}}\|_{L_{1}}\\
    =&\sum_{t}\iint |p_{\theta, f}(\bm{x},t,y)-p_{\theta_{0}, f_{0}}(\bm{x},t,y)|{\rm d}y{\rm d}\bm{x}.
\end{align*}
Let $\mathbb{P}_{0}^{n}$ be the true distribution of sample $\mathcal{D}_{n} =(\bm{t}_n, \bm{X}_n, \bm{y}_n)$, and $\Pi$ be the prior distribution of $p_{\theta, f}$ when the parameters of $\theta$ and $f$ are distributed according to priors $\Pi_{\tau_{\theta}}, \Pi_{\tau_{f}}, \Pi_{\lambda_{\theta}}$, and $\Pi_{\lambda_{f}}$. Then the following theorem holds.
\begin{thm}\label{consitency_theorem}
Suppose that Assumptions (\textbf{P}), (\textbf{F}), (\textbf{T}), (\textbf{E}), and (\textbf{B}) hold.
Then for any $\epsilon>0$,
\begin{align}
    \Pi\left((\theta, f):\|p_{\theta,f}-p_{\theta_{0},f_{0}}\|_{L_{1}}>\epsilon\ |\ \mathcal{D}_{n}\right)\to 0
\end{align}
with $\mathbb{P}_{0}^{n}$-probability $1$.
\end{thm}
\begin{proof}
    See Appendix C for the formal proof.
    Here, we provide a proof sketch.
    From Assumptions (\textbf{P}),  (\textbf{F}),  (\textbf{T}), and  (\textbf{E}),
    we can show that the probability that joint density $p_{\theta,f}$ is not a smooth function is exponentially small.
    Hence, we only have to consider a set of smooth joint density functions.
    With the above four assumptions, 
    we can bound the \textit{metric entropy} of this smooth function set. 
    This allows us to bound the covering number
    and to upper bound the probability of the event that 
    $p_{\theta, f}$ does not enter the neighborhood of $p_{\theta_{0}, f_{0}}$, 
    using a union bound.
    Finally, using Assumption (\textbf{B}), we can bound 
    the Kullback-Leibler divergence between 
    $p_{\theta, f}$ and $p_{\theta_{0},f_{0}}$,
    which enables us to bound the $L_1$ metric, using Pinsker's inequality.    
\end{proof}

\section{Related Work}

\subsection{Bayesian Approaches to CATE Estimation}

The key idea of our method is to put a Gaussian process prior 
on function $\theta(\bm{X})$ in \eqref{proposed_model}, 
which represents how the CATE varies with the values of features $\bm{X}$. 
While \citet{alaa2018bayesian} also employ the Gaussian process priors, 
they propose to place them on 
each mean potential outcome function,
which is denoted by $f_t(\cdot)$ in their model formulation:
\begin{align*}
    Y^{(t)}=f_{t}(\bm{x})+\varepsilon, \quad t\in\left\{0,1\right\}.
\end{align*}
Compared with this model, ours has two advantages. 
First, as described in Section \ref{sec:intro}, 
it can incorporate prior knowledge about the CATE 
by designing covariance function $C(\cdot, \cdot, \bm{\omega}_\theta)$ 
in \eqref{GP-theta};
we provide a formulation example in \eqref{modified_gaussian_kernel},
which can utilize the prior knowledge about treatment effect modifiers.
Second, it can deal with the continuous-valued treatment setup;
hence, the scope of applications is wider.

To develop a Bayesian approach under the continuous-valued treatment setup, 
several methods use a tree-based model called Bayesian additive regression trees (BART)  \citep{hill2011bayesian, woody2020estimating, hahn2020bayesian}.
As reported by \citet{https://doi.org/10.48550/arxiv.1707.02641}, 
these methods empirically work well on many synthetic benchmark datasets.
However, due to the complexity of tree-based models,
their performance is not theoretically guaranteed,
demonstrating that 
it is uncertain whether they can be used for the crucial applications 
that involve critical decision-making.
By contrast, using the Gaussian process priors,
we have derived the asymptotic consistency of posterior distribution,
thus yielding more reliable CATE estimation results.
Furthermore, as illustrated in \eqref{posterior}, 
we can analytically compute the posterior 
when given the hyperparameter values.
Thus, our method requires a much smaller computation time,
compared with the tree-based methods,
which rely on the computationally demanding approximate Bayesian inference.

\subsection{Non-Bayesian Approaches with Partially Linear Model}
\label{subsec:related-NonBayesian}

As described in Section \ref{subsec:PLM}, 
the partially linear model in \eqref{proposed_model}
has been studied in the DML-based methods \citep{chernozhukov2018double},
which are founded on the Frequentist approach. 
Closest to our work is 
the R Learner \citep{nie2021quasi}, 
which also makes the assumption that function $\theta$ in \eqref{proposed_model} 
is an element of RKHS 
(i.e., Assumption (\textbf{T}) in Section \ref{sec:theory}).
For this reason, our method can be regarded as 
the Bayesian counterpart of the R Learner.

A large difference between the DML-based methods and ours is
how to remove the confounding bias
arising from the confounders in features $\bm{X}$. 
To achieve this, 
the DML-based methods use conditional distribution $p(t|\bm{x})$ 
(a.k.a., propensity score)
\footnote{For continuous-valued treatment $t \in \mathbb{R}$,
it is called a \textit{generalized propensity score} \citep{imbens2000role}.},
which is expressed with the following model:
\begin{align}
    T=\rho(\bm{X})+\eta,\quad \mathbb{E}[\eta|\bm{X}]=0. \label{p_t_x}
\end{align}
They estimate function $\rho$ in \eqref{p_t_x}
and use the estimated function to eliminate 
the estimation bias due to nuisance parameter $f$ in the partially linear model \eqref{proposed_model}.
By contrast, with our Bayesian approach, 
we do not need to estimate the propensity score from the data
because we can remove the confounding bias, 
simply by marginalizing out function $f$.
As explained by \citet{li2023bayesian},
the main reason is that it is common in Bayesian inference 
to model the priors such that their parameters are mutually independent.
This independence relation implies that
the parameters of propensity score $\rho$ are conditionally independent of functions
$(\theta, f)$ conditioned on sample $\mathcal{D}_n=(\bm{t}_n, \bm{X}_n, \bm{y}_n)$.
Thus, the propensity score model does not affect the inference,
and we do not need to estimate the propensity score.

A serious disadvantage of the DML-based methods is that 
their uncertainty estimate is founded on a confidence interval, which has no theoretical guarantee in a small sample size setting
\citep{van2000asymptotic}.
As claimed in Section \ref{sec:intro}, 
this disadvantage is fatal
if we focus on the applications 
related to decision-making under uncertainty. 
To overcome this disadvantage, 
we have established a Bayesian inference framework 
that can effectively quantify the CATE estimation uncertainty
in the finite sample size regime.

\section{Experiments}
\label{sec:exp}

\begin{figure*}[t]
    \centering
    \includegraphics[width=\linewidth]{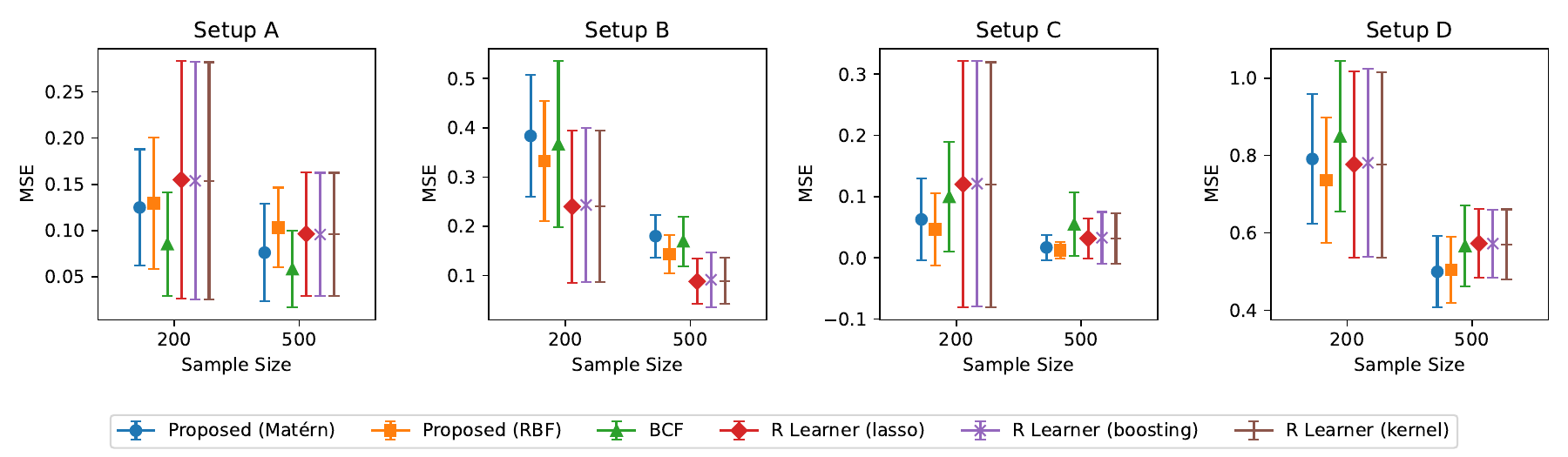}
    \includegraphics[width=\linewidth]{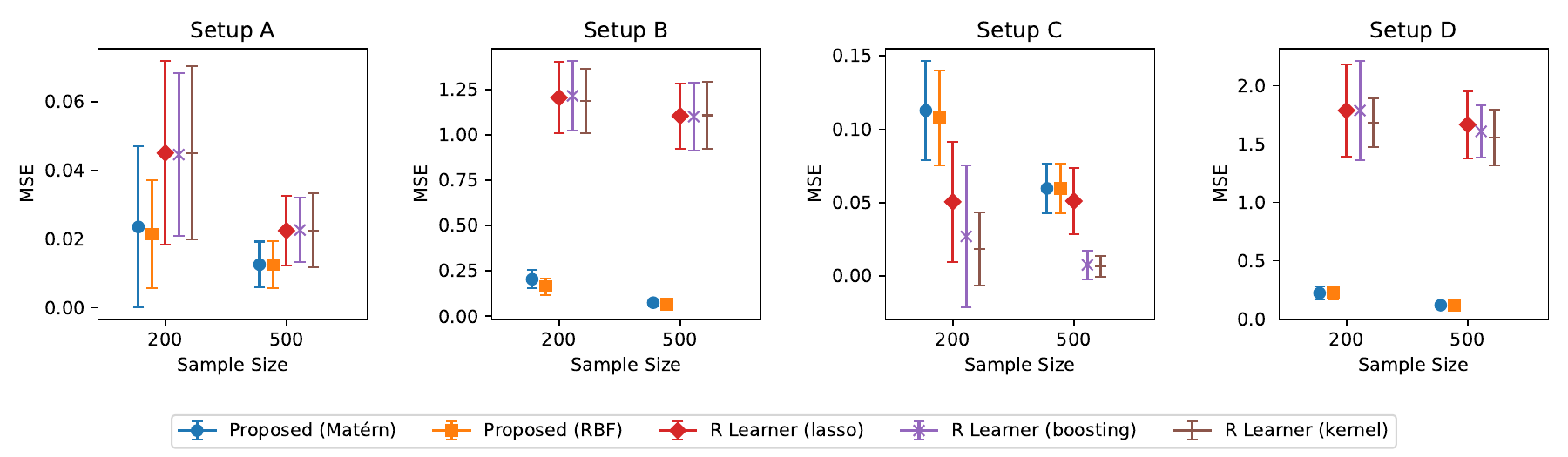}
    \caption{MSEs on synthetic datasets. (Top): binary treatment setup; (Bottom): continuous-values treatment setup. Lower is better.}
    \label{fig:result_synthetic}
\end{figure*}
\begin{figure*}[t]
    \centering
    \includegraphics[width=\linewidth]{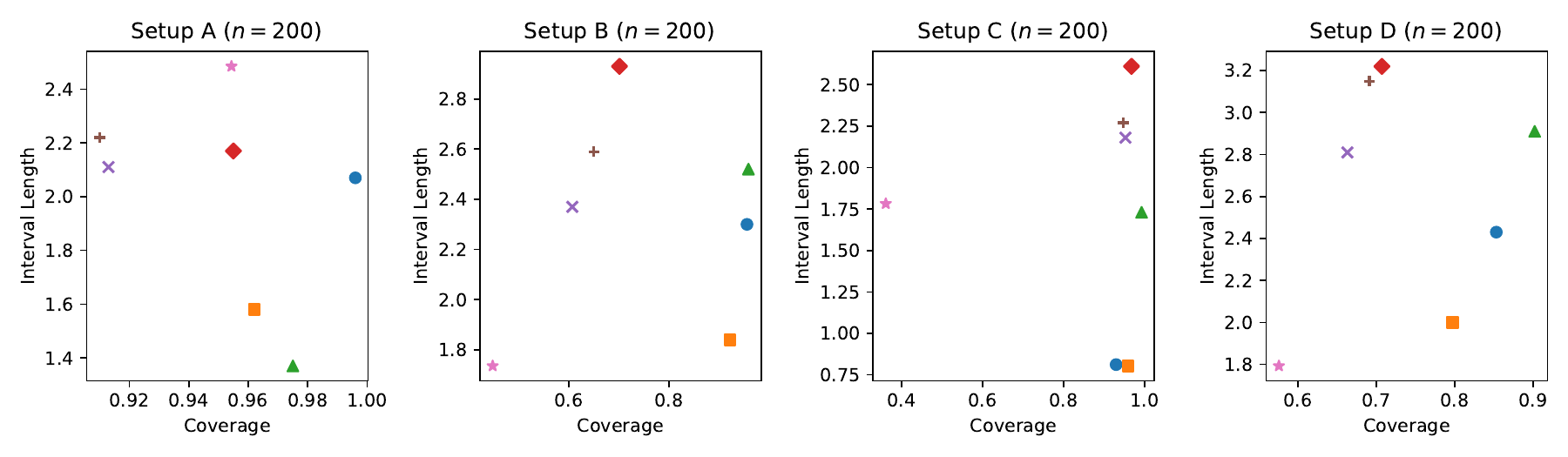}
    \includegraphics[width=\linewidth]{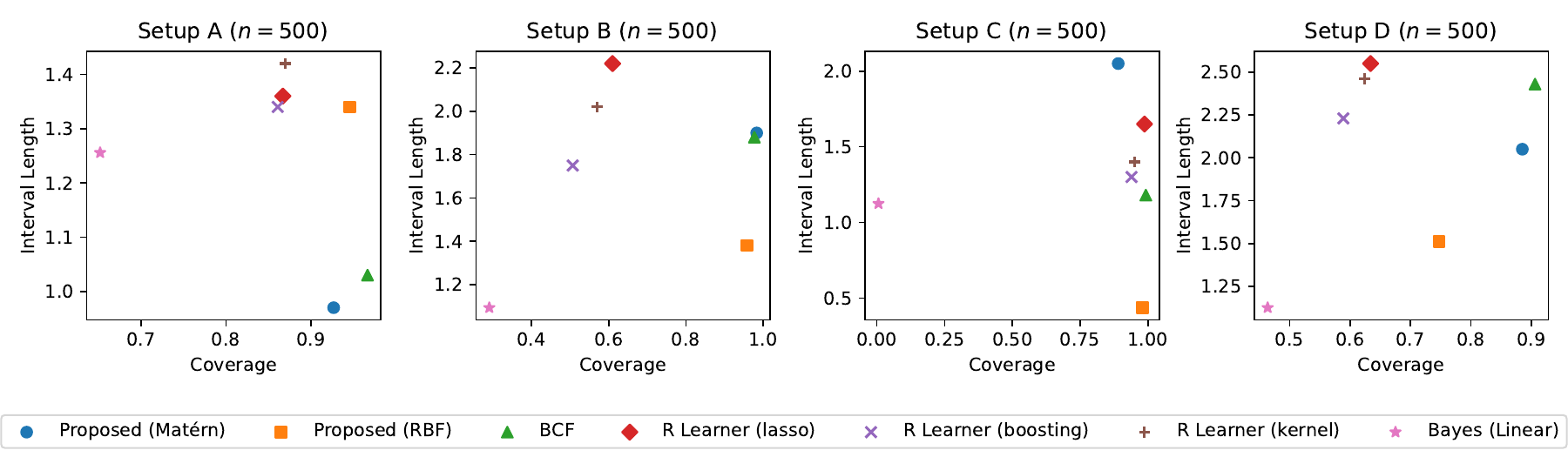}
    \caption{Coverage ratio and interval length of credible/confidence intervals on synthetic data under binary treatment setup. (Top): the sample size $n=200$; (Bottom): the sample size $n=500$. Methods closer to the bottom right corner are better.}
    \label{fig:result_synthetic_binary_interval}
\end{figure*}

Since we have no access to the true CATE values with real-world data,
we evaluated the performance of our method with synthetic and semi-synthetic data.\\
\textbf{Synthetic data}: As with \citet{nie2021quasi}, 
we prepared synthetic data that follow the partially linear model \eqref{proposed_model}.

 We generated four datasets with binary treatment 
using Setup A, B, C, and D in \citet{nie2021quasi}.
In all setups, the number of features $\bm{X}$ is $d = 6$.
Each setup provides different formulations of distributions $p(\bm{x})$ and $p(t|\bm{x})$
and functions $\theta(\bm{X})$ and $f(\bm{X})$ in \eqref{proposed_model}.
The main difference lies in functions $\theta(\bm{X})$ and $f(\bm{X})$:
smooth $\theta$ and $f$ (Setup A);
smooth $\theta$ and non-differentiable $f$ (Setup B);
constant $\theta$ and smooth $f$ (Setup C);
and non-differentiable $\theta$ and $f$ (Setup D).
We detail these setups in Appendix D.1.

As regards continuous-valued treatment setup, 
we modified the above four setups in \citet{nie2021quasi}
to generate the values of treatment $T \in \mathbb{R}$.
In particular, we considered its data-generating process $T = \rho(\bm{X}) + \eta$ (i.e., \eqref{p_t_x})
and formulated function $\rho$ in a different way:
the linear function (Setup A), the constant function (Setup B), 
and the nonlinear and non-differentiable function (Setups C and D).\\
\textbf{Semi-synthetic data}: For binary treatment setup, 
we used the Atlantic Causal Inference Conference (ACIC) dataset 
\citep{shimoni2018benchmarking}, including $1000$ observations
($514$ of whom are treated and $486$ are untreated).
The data of $d=177$ features come from the Linked Birth and Infant Death Data (LBIDD) \citep{macdorman1998infant},
while those of treatment and outcome are simulated. 

Unfortunately, there is no well-established benchmark dataset for the continuous-valued treatment setup,
unlike the binary treatment setup. 
For this reason, 
we focus only on synthetic data experiments for continuous-valued cases. \\
\textbf{Baselines}: We compared our method with the three baselines: the Bayesian causal forest (BCF) method \citep{hahn2020bayesian},\footnote{\url{https://github.com/jaredsmurray/bcf}} the R Learner \citep{nie2021quasi},\footnote{
    We used the original implementation in \url{https://github.com/xnie/rlearner} for the binary treatment setup.
    As regards continuous-valued treatment setup,
since there is no original implementation,
we employed the KernelDML class in the EconML package downloaded from \url{https://econml.azurewebsites.net/index.html}.
} which employs a kernel function to infer function $\theta$ in the partially linear model \eqref{proposed_model}, and the Bayesian linear regression model.
In Appendix F, we present the comparison with the additional baseline, the stableCFR method \cite{wu2023stable}, which is a recent neural-network-based method.

As regards the continuous-valued treatment setup,
we compared our method only with the R Learner
because the original implementation of BCF is unavailable.

To examine the advantages of using a partially linear model, we also compared with the Bayes optimal estimator when employing a linear model defined as $Y=(\bm{\beta}_{\theta}^{\top}\bm{X})T+\bm{\beta}_{f}^{\top}\bm{X}+\varepsilon$.
In this case, the Bayes optimal estimator of the CATE is given by the posterior mean of $\bm{\beta}_{\theta}^{\top}\bm{X}$.\\
\textbf{Evaluation measures}: 
To measure the CATE estimation performance, 
we conducted $100$ experiments and
computed the average and standard deviation of 
the mean squared error (MSE): 
$\frac{1}{m} \sum_{i=1}^m (\hat{\theta}(\tilde{\bm{x}}_i) - \theta(\tilde{\bm{x}}_i))^2$,
where $\hat{\theta}(\cdot)$ denotes the CATE estimated with training data,
and $\tilde{\bm{x}}_1, \dots, \tilde{\bm{x}}_m \overset{i.i.d.}{\sim} p(\bm{x})$ are the feature values in test data.
In all experiments, we used $n=200$ or $n=500$ observations as training data
and $m=100$ observations as test data.
With the ACIC dataset, 
the true CATE value, $\theta(\bm{x}_i)$, is given by a difference between the simulated potential outcomes,
and we evaluated the MSE by randomly selecting training and test data from $1000$ observations.
Note that under the continuous-valued treatment setup, CATE $\theta(\bm{x})$ in the MSE corresponds to $(t'-t) \theta(\bm{x})$ in \eqref{eq-CATEs} when $t' = t + 1$.

To examine the performance of uncertainty estimation,
we computed the 95\% credible interval for the Bayesian methods 
(i.e., the proposed method and BCF)
and the 95\% confidence interval for the Frequentist-based method 
(i.e., the R Learner).
To measure the quality of these intervals,
we computed the coverage ratio 
(i.e., the ratio in which the true CATE value is included)
and the length of the interval.\\
\textbf{Results: }
Figure \ref{fig:result_synthetic} presents 
the MSEs of each method on synthetic datasets.
Here, we omitted the results with the Bayesian linear regression model due to its extremely large MSE values compared with other methods.
See Appendix \ref{sec:rebuttal_experiment} for the results.
 
Our method achieved better or more competitive performance in binary and continuous-valued treatment setups than the BCF and the R Learner.
Our method worked well
when the data were generated from the partially linear model with non-differentiable functions (i.e., Setups B and D).
In particular,
our method outperformed the other two baselines under Setup D. 
These results demonstrate that 
even if function $\theta(\cdot)$, which represents the CATE, is non-differentiable,
our method can approximate it with the smooth functions 
induced by the kernel function.
One of the reasons why our method can perform such effective inference 
is that it does not rely on any approximate posterior computation,
unlike the BCF method.
The proposed method performs worse than other methods in Setup C of the continuous-values treatment setup.
Our method failed to approximate the CATE function $\theta$, which was given as a constant function.
This is because it can be difficult to approximate such a too-simple function with smooth functions in a small sample size setting.

With the ACIC dataset, we obtained similar results 
(Figure \ref{fig:result_LBIDD}).
We observed that 
the performance of the R Learner strongly depended on 
the choice of regression models for estimating the conditional expectations, 
$\mathbb{E}[T|\bm{X}=\bm{x}]$ and $\mathbb{E}[Y|\bm{X}=\bm{x}]$.
When using kernel regression (i.e., R Learner (kernel); omitted in Figure \ref{fig:result_LBIDD}), 
the MSEs were extremely large:
$71.8 \pm 11.0$ ($n=200$) and $52.3\pm 4.3$ ($n=500$).
Such unstable performance may raise serious doubt for practitioners
because appropriately selecting the regression model requires 
deep understanding of the data analysis.
By contrast, our method worked well,
regardless of the choice of kernel functions (i.e., Mat\'ern and RBF kernels).

Figure \ref{fig:result_synthetic_binary_interval} shows the performance of uncertainty quantification on synthetic data.
A higher coverage ratio and a shorter interval length are better: the methods closer to the bottom right corner exhibit better performance.
Especially when $n$ is small, 
our method and BCF achieved larger coverage ratios than the R Learner 
while keeping the length small,
thus demonstrating the effectiveness of the Bayesian approaches.

\begin{figure}[t]
    \centering
    \includegraphics[width=\linewidth]{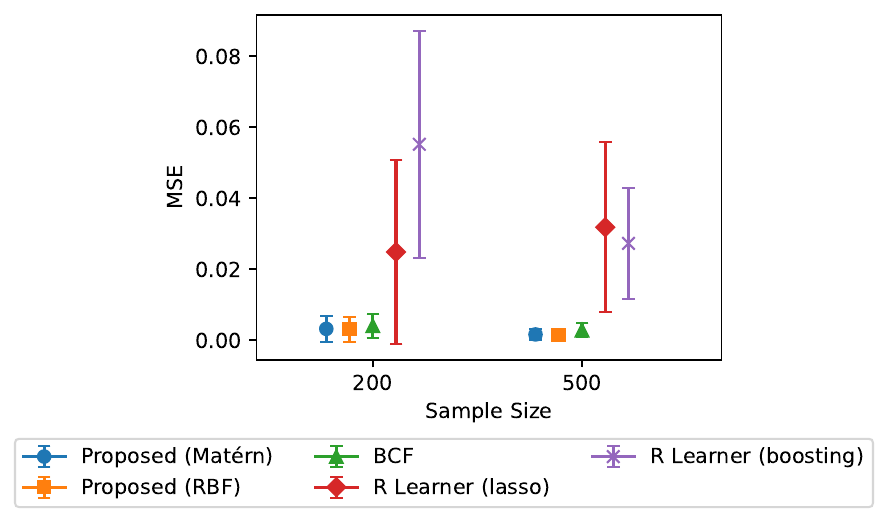}
    \caption{MSEs on ACIC dataset. MSEs of R Learner (kernel) are omitted due to their extremely large values.}
    \label{fig:result_LBIDD}
\end{figure}

\section{Conclusion}

We proposed a Bayesian framework that quantifies the CATE estimation uncertainty with the posterior distribution.
The key idea is to put Gaussian process priors on the nonparametric components in a partially linear model.
This idea offers a computationally efficient Bayesian inference 
with a closed-form posterior of the CATE.
Moreover, it enables us to incorporate prior knowledge about the CATE,
leading to an effective posterior inference, as empirically demonstrated in our experimental results (Appendix E).

Theoretically, we prove that the posterior has asymptotic consistency under some mild conditions.
Our future work constitutes further investigation.
In particular, we will investigate the relationship 
between the minimax information rate and the assumed class of nonlinear functions,
as with the results of the Gaussian-process-based model \citep{alaa2018bayesian}.

\newpage

\section*{Acknowledgments}

This research is partially supported by the Telecommunications Advancement Foundation, and No. 22K12156 of Grant-in-Aid for Scientific Research Category (C), Japan Society for the Promotion of Science.

\bibliography{ref}

\newpage
\ 
\newpage

\appendix
\renewcommand{\thetable}{A.\arabic{table}}
\renewcommand{\theequation}{A.\arabic{equation}}

\section{Conditional Derivative Effect and Partially Linear Model}

In Section 2.2, we show that 
the CATE is expressed using function $\theta$ in the partially linear model: 
\begin{align*}
    Y=\theta(\bm{X}) T+f(\bm{X})+\varepsilon, 
\end{align*}
This section illustrates that 
other treatment effect measure called a \textit{conditional derivative effect} (CDE)
is also formulated using the function $\theta$.

As with the CATE, a CDE is an average treatment effect 
across individuals with features $\bm{X}=\bm{x}$ \citep{hines2021parameterising}.
Unlike the CATE, however, it focuses only on continuous-valued treatment $t \in \mathbb{R}$
and measures an average treatment effect as an average derivative 
(a.k.a., \textit{average derivative effect} \citep{hardle1989investigating}):
\begin{align}
    \mathrm{CDE}(\bm{x})= \mathrm{lim}_{\xi \rightarrow 0} \frac{1}{\xi} \mathbb{E}[Y^{(t+\xi)} - Y^{(t)} | \bm{X}=\bm{x}]. \label{CDE}
\end{align}

Under our partially linear model formulation, 
since the CDE corresponds to the derivative with respect to treatment $T$,
it is formulated as
\begin{align}
    \mathrm{CDE}(\bm{x})= \theta(\bm{x}). \label{CDE-PLM}
\end{align}

Hence, under the continuous-valued treatment setup,
the posterior inference of function $\theta$ in a partially linear model also allows us to estimate the CDE.

\section{Assumptions for Theorem 1}

To state our assumptions, we first introduce the notations about the prior distributions. Let the covariance functions in our Gaussian process priors on  $\theta$ and $f$ be
\begin{align*}
    C(\bm{x}, \bm{x}'; \tau_{\theta}, \lambda_{\theta})&=\tau_{\theta}^{-1}k_{0}(\lambda_{\theta}\bm{x}, \lambda_{\theta}\bm{x}'),\\
    C(\bm{x}, \bm{x}'; \tau_{f}, \lambda_{f})&=\tau_{f}^{-1}k_{0}(\lambda_{f}\bm{x}, \lambda_{f}\bm{x}'),
\end{align*}
where $k_{0}(\cdot,\cdot)$ is a nonsingular covariance kernel, and $\tau_{\theta}$,$\tau_{f}$, $\lambda_{\theta}, \lambda_{f} \geq 0$ are hyperparameters. Let the priors on these hyper-parameters be $\Pi_{\tau_{\theta}}, \Pi_{\tau_{f}}, \Pi_{\lambda_{\theta}}$ and $\Pi_{\lambda_{f}}$, respectively.
\footnote{As noted in Section 3.3,
in our experiments, we did not put the priors on these hyperparameters
 due to the computational complexity. 
 Thus, there is a gap between the theoretical analysis and the experimental results. Extending the theoretical results is left as our future work;
 recent results on the analysis of the empirical Bayes approach 
 in nonparametric regression might be helpful \citep{szabo2015frequentist}.}
For sample size $n$ and for some constant $c$, let $\lambda_{\theta,n}, \lambda_{f,n}, \tau_{\theta,n}$, and $\tau_{f,n}$ be the sequences that satisfy 
\begin{align*}
    &\Pi_{\tau_{\theta}}(\tau_{\theta}<\tau_{\theta,n})=e^{-cn}; \quad \Pi_{\lambda_{\theta}}(\lambda_{\theta}>\lambda_{\theta,n})=e^{-cn}\\
    &\Pi_{\tau_{f}}(\tau_{f}<\tau_{f,n})=e^{-cn}; \quad \Pi_{\lambda_{f}}(\lambda_{f}>\lambda_{f,n})=e^{-cn}.
\end{align*}

Next, we introduce the notations about joint density $p_{\theta,f}(\bm{x}, t, y)$.
We define the subset of joint densities as
\begin{align*}
    \mathcal{P}_{n, \alpha}&=\left\{p_{\theta, f}: \theta, f\in\mathcal{G}_{n,\alpha}\right\},
\end{align*}
where $\mathcal{G}_{n,\alpha}$ denotes the set of smooth functions: 
\begin{align*}
    \mathcal{G}_{n,\alpha}&=\left\{g:\|D^{w}g\|_{\infty}<M_{n}, w\le \alpha \right\},
\end{align*}
where $D^{w}g = (\partial^{w}/\partial^{w_{1}}\ldots\partial^{w_{d}})g(x_{1},\ldots,x_{d})$ is a partial derivative defined 
with positive integers $w_1, \dots, w_d$ and their sum $w=\sum_{i=1}^d w_{i}$,
$\alpha$ is a positive integer, and $M_{n}$ is a sequence of real numbers.

Based on the above notations, we make the assumptions:
\begin{description}
    \item[(P)] (\textbf{Smoothness of priors})
    For every fixed $\bm{x}\in\mathbb{R}^d$, covariance kernel $k_{0}(\bm{x},\cdot)$ has continuous partial derivatives up to order $2\alpha+2$, where $\alpha$ is a positive integer which satisfies the condition described in Assumption (\textbf{E}). Priors $\Pi_{\lambda_{\theta}}$ and $\Pi_{\lambda_{f}}$ are fully supported on $(0, \infty)$.
    \item[(F)] (\textbf{Bounded feature space})
    Feature vector values $\bm{x}$ belong to a bounded subset of $\mathbb{R}^{d}$.
    \item[(T)] (\textbf{True functions})
    True functions $\theta_{0}$ and $f_{0}$ belong to the reproducing kernel Hilbert space (RKHS) of $k_{0}$.
    \item[(E)] (\textbf{Exponential decay of priors})
    For every $b_{1}>0$ and $b_{2}>0$, there exist sequences $M_{n}, \lambda_{\theta,n}, \lambda_{f,n}, \tau_{\theta,n}$ and $\tau_{f,n}$ that satisfy
    \begin{align*}
        M_{n}^{2}\tau_{\theta,n}\lambda_{\theta,n}^{-2}\ge b_{1}n,\quad
        M_{n}^{d\alpha}\le b_{2}n,\\
        M_{n}^{2}\tau_{f,n}\lambda_{f,n}^{-2}\ge b_{1}n,\quad
        M_{n}^{d\alpha}\le b_{2}n.
    \end{align*}
\end{description}
These assumptions correspond to the modified ones of the existing results 
\citep{ghosal1999posterior}, 
which proves the posterior consistency 
for the standard Gaussian process regression problems.

In addition, as described in Section 4,
we make an additional assumption on the boundedness of 
the conditional moments of treatment $T$ given features $\bm{X}$:
\begin{description}
    \item[(B)] (\textbf{Boundedness of conditional moments})
    There exist constants $C_{1}>0$ and $C_{2}>0$ such that
    \begin{align*}
        \mathbb{E}\left[T|\bm{X}\right]<C_{1}\quad \mathbb{P}_{\bm{X}}\textrm{-}a.s.;\quad
        \mathbb{E}\left[T^{2}|\bm{X}\right]<C_{2}\quad \mathbb{P}_{\bm{X}}\textrm{-}a.s.
    \end{align*}
\end{description}
This assumption imposes the conditional mean and variance of $T$ given $\bm{X}$ to be at most $C_1$ and $C_2$, respectively.

\section{Proof of Theorem 1}
From Theorem 2 in \citet{ghosal1999posterior}, it suffices to verify the following conditions hold:
\begin{itemize}
    \item $\Pi\left((\theta, f):\mbox{KL}(p_{\theta_{0},f_{0}}||p_{\theta,f})<\epsilon\right)>0$ for every $\epsilon>0$, where $\mbox{KL}$ is the Kullback-Leibler (KL) divergence.
    \item There exists $\beta>0$ such that $\log N(\epsilon, \mathcal{P}_{n,\alpha}, ||\cdot||_{1})<n\beta$, where $N(\epsilon, \mathcal{P}_{n,\alpha}, ||\cdot||_{1})$ is the covering number (and its logarithm is the metric entropy).
    \item $\Pi(\mathcal{P}^{c}_{n,\alpha})$ is exponentially small.
\end{itemize}

We first show that $\Pi\left((\theta, f):\mbox{KL}(p_{\theta_{0},f_{0}}||p_{\theta,f})<\epsilon\right)>0$ for every $\epsilon>0$.
From the definition of KL divergence,
\begin{align}
    &\mathrm{KL}(p_{\theta_{0},f_{0}}||p_{\theta,f})\nonumber \\
    =&\mathbb{E}_{\bm{X},T,Y}\left[\log \frac{p_{\theta_{0},f_{0}}(\bm{X}, T, Y)}{p_{\theta,f}(\bm{X}, T, Y)}\right]\nonumber \\
    =&\mathbb{E}_{\bm{X},T}\left[\mathbb{E}_{Y}\left[\log \frac{p(Y|\bm{X},T, \theta_{0},f_{0})}{p(Y|\bm{X},T,\theta,f)}|\bm{X},T\right]\right],\label{sum_KL}
\end{align}
where the expectations are taken with respect to distribution $p_{\theta_{0},f_{0}}$.
When $Y=\theta_{0}(\bm{X})T+f_{0}(\bm{X})+\varepsilon$ and $\mathbb{E}\left[\varepsilon\right]=0$, $\mathbb{E}_{Y}\left[Y|\bm{X},T\right]=\theta_{0}(\bm{X})T+f_{0}(\bm{X})$.
Substituting this equation into (\ref{sum_KL}), some algebra leads to
\begin{align}
    &\mathrm{KL}(p_{\theta_{0},f_{0}}||p_{\theta,f})\nonumber \\
    =&\mathbb{E}_{\bm{X},T}\Bigl[\frac{s_{\varepsilon}}{2}(\theta_{0}\left((\bm{X})T+f_{0}(\bm{X}))-\right.\nonumber\\
    &\hspace{3cm} \left.(\theta(\bm{X})T+f(\bm{X}))\right)^{2}|\bm{X},T\Bigr]\nonumber \\
    =&\mathbb{E}_{\bm{X},T}\Bigl[\frac{s_{\varepsilon}}{2}\left((\theta_{0}(\bm{X})-\theta(\bm{X}))^{2}T^{2}+\right.\nonumber\\
    &\qquad \qquad  2(\theta_{0}(\bm{X})-\theta(\bm{X}))(f_{0}(\bm{X})-f(\bm{X}))T+\nonumber\\
    &\left.\hspace{3.8cm}(f_{0}(\bm{X})-f(\bm{X}))^{2}\right)|\bm{X},T\Bigr]\nonumber\\  &=\mathbb{E}_{\bm{X}}\Biggl[\mathbb{E}_{T}\Bigl[\frac{s_{\varepsilon}}{2}\left((\theta_{0}(\bm{X})-\theta(\bm{X}))^{2}T^{2}+\nonumber\right.\\
    &\hspace{1.5cm}2(\theta_{0}(\bm{X})-\theta(\bm{X}))(f_{0}(\bm{X})-f(\bm{X}))T+\nonumber\\
    &\left.\hspace{4cm}(f_{0}(\bm{X})-f(\bm{X}))^{2}\right)|\bm{X}\Bigr]\Biggr]\nonumber\\
    \le &\mathbb{E}_{\bm{X}}\Biggl[\frac{s_{\varepsilon}}{2}\left(C_{2}(\theta_{0}(\bm{X})-\theta(\bm{X}))^{2}+\nonumber \right.\\
    &\hspace{1cm}2C_{1}(\theta_{0}(\bm{X})-\theta(\bm{X}))(f_{0}(\bm{X})-f(\bm{X}))+\nonumber\\
    &\left.\hspace{4cm}(f_{0}(\bm{X})-f(\bm{X}))^{2}\right)|\bm{X}\Biggr]\nonumber\\
    \le &\frac{s_{\varepsilon}}{2}\left(C_{2}\|\theta_{0}-\theta\|_{\infty}^{2}+\right.\nonumber\\
    &\hspace{1cm}\left.2C_{1}\|\theta_{0}-\theta\|\|f_{0}-f\|+\|f_{0}-f\|_{\infty}^{2}\right),\label{KL_bound}
\end{align}
where the first inequality follows from the assumption that $\mathbb{E}[T|\bm{X},T]<C_{1}, \mathbb{E}[T^{2}|\bm{X},T]<C_{2},a.s.$ and the second inequality follows from Cauchy-Schwarz inequality.
From Theorem 4 in \citep{ghosal1999posterior}, it holds $\Pi(\theta: ||\theta_{0}-\theta||_{\infty}<\epsilon)>0$ and $\Pi(f:||f_{0}-f||_{\infty}<\epsilon)>0$, thus $\Pi((\theta,f):\mbox{KL}(p_{\theta_{0},f_{0}}||p_{\theta,f})<\epsilon)>0$.

Next, we bound the covering number $N(\epsilon, \mathcal{P}_{n, \alpha}, ||\cdot||_{1})$.
To simplify the description, we write $N(\epsilon, \mathcal{G}_{n}, ||\cdot||_{\infty})$ as $N_{\epsilon}$.
From the definition of the covering number, we can construct $\theta_{1},\ldots,\theta_{N_{\epsilon}}$, $f_{1},\ldots,f_{N_{\epsilon}}$ that satisfy the following condition:
\begin{itemize}
    \item There exists $i,j\in\left\{1,\ldots,N_{\epsilon}\right\}$ that satisfy $||\theta-\theta_{i}||_{\infty}<\epsilon, ||f-f_{i}||_{\infty}<\epsilon$ for all $\theta, f\in\mathcal{G}_{n}$.
\end{itemize}
We choose $\theta,f\in\mathcal{G}_{n}$ arbitrarily and let $\theta^{*}, f^{*}$ be the functions that satisfy the above condition.
Then it holds
\begin{align}
    &||p_{\theta,f}-p_{\theta^{*},f^{*}}||_{1}\nonumber\\
    \le &\sqrt{2\mbox{KL}(p_{\theta,f}||p_{\theta^{*},f^{*}})}\nonumber\\
    \le &\left(s_{\varepsilon}\left(C_{2}\|\theta-\theta^{*}\|_{\infty}^{2}+\right.\right. \nonumber\\
    &\left. \left.2C_{1}\|\theta-\theta^{*}\|_{\infty}\|f-f^{*}\|_{\infty}+\|f-f^{*}\|_{\infty}^{2}\right)\right)^{1/2}\nonumber\\
    \le &\sqrt{s_{\varepsilon}(2C_{1}+C_{2}+1)}\epsilon, \nonumber
\end{align}
where the first line follows from Pinsker's inequality \citep{Cover2006}, the second line follows from (\ref{KL_bound}), and the last line follows from the definition of $\theta^{*},f^{*}$.
Let $C=2C_{1}+C_{2}+1$.
The above inequality implies $N(\sqrt{Cs_{\epsilon}}\epsilon, \mathcal{P}_{M_{n},\alpha},||\cdot||_{1})\le N_{\epsilon}^{2}$.
From the proof of Theorem 1 in \citep{ghosal1999posterior}, it holds $\log N(\epsilon, \mathcal{G}_{n}, ||\cdot||_{\infty})\le K\epsilon^{-d/\alpha}b^{d/\alpha}n$, so
\begin{align*}
    \log N(\epsilon, \mathcal{P}_{M_{n},\alpha}, ||\cdot||_{1})\le 2K(\epsilon/\sqrt{Cs_{\epsilon}})^{-d/\alpha}b^{d/\alpha}n.
\end{align*}
By letting $b<(\beta/(2K))^{\alpha/d}(\epsilon/\sqrt{Cs_{\epsilon}})$, it holds 
\begin{align*}
\log N(\epsilon, \mathcal{P}_{n}, \left||\cdot|\right|_{1})<n\beta.    
\end{align*}

Finally, it is easy to verify that $\Pi(\mathcal{P}_{M_{n},\alpha}^{c})$ is exponentially small from Lemma 1 in \citep{ghosal1999posterior}.

\section{Experimental Settings}

This section details the settings of the experiments described in Section 6. 
First, we provide the formulations of data generation processes (DGPs) used in synthetic data experiments.
Then we describe the parameter settings of each method.

\subsection{Synthetic Data}

\textbf{Binary Treatment Setup}:
In Section 6.1,
we used the four DGPs (Setup A, B, C, and D) introduced by \citet{nie2021quasi}.
All DGPs generate the $i$-th observation ($i = 1, \dots, n$) as
\begin{align}
    &\bm{X}_{i}\sim p(\bm{x})  \label{Nie-p_x} \\
    &T_{i}|\bm{X}_{i}\sim \mathrm{Bernoulli}(e(\bm{X}_{i})) \label{Nie-p_t_x} \\
    &Y_{i}=\theta(\bm{X}_{i})T_{i} + f(\bm{X}_i) + \varepsilon_{i} \quad (\varepsilon_i \sim \mathcal{N}(0, 1)) \label{Nie-PLM},
\end{align}
where $\bm{X}_i \in \mathbb{R}^d$ ($d = 6$).
\citet{nie2021quasi} formulate function $f$ in the partially linear model 
 as
\begin{align}
    f(\bm{X}_i) = b(\bm{X}_{i})- 0.5 \theta(\bm{X}_{i}) \label{Nie-f}.
\end{align}

In each DGP, $p(\bm{x})$ and functions $e(\bm{X}), \theta(\bm{X})$, and $b(\bm{X})$ 
in \eqref{Nie-p_x}, \eqref{Nie-p_t_x}, \eqref{Nie-PLM}, and \eqref{Nie-f}
are differently formulated. 
Setup A uses smooth functions $\theta$ and $f$:
\begin{align*}
    &p(\bm{x}) = \mathrm{Unif}(0, 1)^{d}, \\
    &e(\bm{X}) = \mathrm{trim}_{0.1}(\sin(\pi X_{1} X_{2})),\\
    &\theta(\bm{X})= \frac{1}{2}(X_{1}+ X_{2}),\\
    &b(\bm{X})=\sin(\pi X_{1}X_{2})+2(X_{3}-0.5)^{2}+X_{4}+0.5X_{5}, 
\end{align*}
where $\mathrm{trim}_{\eta}(x)=\max\left\{\eta, \min(x, 1-\eta)\right\}$.
Setup B employs smooth $\theta$ and non-differentiable $f$, namely
\begin{align*}
    &p(\bm{x}) = \mathcal{N}(\bm{0}, \bm{I}), \\
    &e(\bm{X}) = 0.5,\\
    &\theta(\bm{X})=X_{1}+\log(1+\mathrm{e}^{X_{2}}),\\
    &b(\bm{X})=\max\left\{X_{1}+X_{2}, X_{3}, 0\right\}+\max\left\{X_{4}+X_{5}, 0\right\},
\end{align*}
where function $e(\bm{X})$ is given as a constant function; hence, it simulates a randomized controlled trial.
Setup C is formulated with constant function $\theta$ and smooth $f$ as
\begin{align*}
    &p(\bm{x}) = \mathcal{N}(\bm{0}, \bm{I}), \\
    &e(\bm{X}) = \frac{1}{1+\mathrm{e}^{X_{2}+X_{3}}},\\
    &\theta(\bm{X})=1,\\
    &b(\bm{X})=2\log(1+\mathrm{e}^{X_{1}+X_{2}+X_{3}}),
\end{align*}
indicating that the true CATE value is $1$ regardless of $\bm{X}$'s values. 
Setup D uses non-differentiable $\theta$ and $f$:
\begin{align*}
    &p(\bm{x}) = \mathcal{N}(\bm{0}, \bm{I}), \\
    &e(\bm{X}) = \frac{1}{1+\mathrm{e}^{-X_{1}-X_{2}}},\\
    &\theta(\bm{X})=\max\left\{X_{1}+X_{2}+X_{3}, 0\right\}-\max\left\{X_{4}+X_{5}, 0\right\},\\
    &b(\bm{X})=\frac{1}{2}(\max\left\{X_{1}+X_{2}+X_{3}, 0\right\}+\\
    &\hspace{5cm}\max\left\{X_{4}+X_{5}, 0\right\}).
\end{align*}
\textbf{Continuous-valued Treatment Setup}: In Section 6.2, we prepared synthetic data 
by modifying the above four DGPs.

To produce the $i$-th value of continuous-valued treatment $T \in \mathbb{R}$ for $i = 1, \dots, n$,
we formulated its data-generating process as
\begin{align}
    T_i = \rho(\bm{X}_i) + \eta_i \quad (\eta_i \sim \mathcal{N}(0, 1)), 
\end{align}
where $\rho\colon \mathbb{R}^d \rightarrow \mathbb{R}$ is a function, and $\eta_i$ denotes a standard Gaussian noise.

We modified the four DGPs in \citet{nie2021quasi} 
by employing different formulations of function $\rho$.
In particular, Setup A uses linear function:
\begin{align*}
    \rho(\bm{X}) = X_1 + X_2,
\end{align*}
Setup B employs constant:
\begin{align*}
    \rho(\bm{X}) = 0,
\end{align*}
and Setups C and D use a non-differentiable function:
\begin{align*}
    \rho(\bm{X}) = \max\left\{X_{1}+X_{2}+X_{3}, 0\right\}-\max\left\{X_{4}+X_{5}, 0\right\}.
\end{align*}
Note that the last non-differentiable function is nonlinear. 
To confirm this, consider the two inputs, $\bm{x} = [1, 0, 0, 0, 0, 0]^{\top}$ and $\bm{x}' = [-1, 0, 0, 0, 0, 0]^{\top}$. Then it holds that $\rho(\bm{x}) = 1 - 0 = 1$ and $\rho(\bm{x}') = 0 - 0 = 0$. However, $\rho(\bm{x}+\bm{x}') = 0 - 0 = 0$. Thus the linearity does not hold; hence, it is a nonlinear function.

\subsection{Parameter Settings}

\textbf{Proposed Method}: To formulate the covariance functions in Gaussian process priors in (6) and (7),
we used the Mat\'ern kernel and the RBF kernel.
As described in Section 3.3, the marginal likelihood is not necessarily convex with respect to the hyperparameters of these kernel functions, $\bm{\omega}_{\theta}$ and $\bm{\omega}_{f}$, and the noise precision parameter in the Gaussian likelihood, $s_{\varepsilon}$.
Thus, we maximize the marginal likelihood with respect to these hyperparameters by combining the grid search and the gradient descent method.
The overview of the algorithm is to determine the optimal value of $s_{\varepsilon}$ for each point in the lists $\Omega_{\theta}$ and $\Omega_{f}$, which are the targets for search of $\bm{\omega}_{\theta}$ and $\bm{\omega}_{f}$, using gradient descent. The algorithm then outputs the $\bm{\omega}_{\theta}$ and $\bm{\omega}_{f}$ that maximize the marginal likelihood value for that $s_{\varepsilon}$. For details, refer to Algorithm 1 and Algorithm 2.
The optimization targets only the scale parameter for both the Mat\'ern kernel and the RBF kernel. For both $\Omega_{\theta}$ and $\Omega_{f}$, the range was set to $[10^{-3}, 10^{-2.5}, \ldots, 10^{2.5}, 10^{3}]$.

\textbf{BCF}: The BCF method takes two steps: 
it first estimates the propensity score 
and then takes as input the estimated propensity score values in function $f$.
As a parametric model of the propensity score, we used a logistic regression model.\\
\textbf{R Learner}:
The R Learner requires the propensity score and the conditional outcome models, 
each of which represents conditional expectations
$\mathbb{E}[T|\bm{X}]$ and $\mathbb{E}[Y|\bm{X}]$, respectively.
Following the original paper \citep{nie2021quasi},
we formulated each conditional expectation
by employing lasso regression, boosting regression, and kernel ridge regression.
We tuned the parameters of each regression model using 5-fold cross-validation.

\begin{algorithm}[t]
\caption{Find Optimal Hyperparameters}
\begin{algorithmic}
\Procedure{find\_opt\_hyper}{$\bm{X}_{n}, \bm{t}_{n}, \bm{y}_{n}, \Omega_{\theta}, \Omega_{f}$, \text{lr}}
    \State $\ell^{*} \gets -\infty$
    \State $\bm{\omega}_{\theta}^{*}, \bm{\omega}_{f}^{*} \gets \text{None}, \text{None}$
    \State $s_{\varepsilon}^{*} \gets \text{None}$
    \For{$\bm{\omega}_{\theta}$ in $\Omega_{\theta}$}
        \For{$\bm{\omega}_{f}$ in $\Omega_{f}$}
            \State $s_{\varepsilon}', \ell' \gets \text{GD\_S}(\bm{X}_{n}, \bm{t}_{n}, \bm{y}_{n}, \bm{\omega}_{\theta}, \bm{\omega}_{f}, \text{lr})$
            \If{$\ell' > \ell^{*}$}
                \State $\ell^{*} \gets \ell'$
                \State $\bm{\omega}_{\theta}^{*}, \bm{\omega}_{f}^{*} \gets \bm{\omega}_{\theta}, \bm{\omega}_{f}$
                \State $s_{\varepsilon}^{*} \gets s_{\varepsilon}'$
            \EndIf
        \EndFor
    \EndFor
    \State \Return $s_{\varepsilon}^{*}, \bm{\omega}_{\theta}^{*}, \bm{\omega}_{f}^{*}$
\EndProcedure
\end{algorithmic}
\end{algorithm}

\begin{algorithm}[t]
\caption{Gradient Descent for $s_{\varepsilon}$}
\begin{algorithmic}
\Procedure{GD\_S}{$\bm{X}_{n}, \bm{t}_{n}, \bm{y}_{n}, \bm{\omega}_{\theta}, \bm{\omega}_{f}$, \text{lr}}
    \State $s_{\varepsilon} \gets n\sum_{i=1}^{n}(y_{i}-\bar{y})^{-2}$
    \State $\ell' \gets \infty$
    \State $i \gets 0$
    \While{$i \leq \text{max\_iter}$}
        \State $g \gets \frac{\partial}{\partial s_{\varepsilon}}\log p(\bm{y}_{n}|\bm{t}_{n}, \bm{X}_{n}; \bm{\omega}_{\theta}, \bm{\omega}_{f}, s_{\varepsilon})$
        \State $s_{\varepsilon} \gets s_{\varepsilon} + \text{lr} \times g$
        \State $\ell \gets \log p(\bm{y}_{n}|\bm{t}_{n}, \bm{X}_{n}; \bm{\omega}_{\theta}, \bm{\omega}_{f}, s_{\varepsilon})$
        \If{$\left| \ell - \ell' \right| < \epsilon$}
            \State \textbf{break}
        \EndIf
        \State $\ell' \gets \ell$
        \State $i \gets i + 1$
    \EndWhile
    \State \Return $s_{\varepsilon}, \ell$
\EndProcedure
\end{algorithmic}
\end{algorithm}

\section{Additional Experiments on Prior Knowledge Incorporation}

In this section, we provide experimental results 
that demonstrate the effectiveness of 
utilizing the prior knowledge about treatment effect heterogeneity,
namely,
 the prior knowledge about 
important features for treatment effect heterogeneity 
(i.e., treatment effect modifiers).
To represent this prior knowledge, we employed the covariance function in (9),
which can take into account the importance of each feature
by setting its hyperparameter values.\\
\textbf{Data}: We considered a binary treatment setup and generated the $i$-th observation 
($i=1,\ldots,n$) by
\begin{align}
    &\bm{X}_{i}\sim \mathcal{N}(\bm{0}, \bm{I}), \label{additional_p_x}\\
    &T_{i}|\bm{X}_{i}\sim \mathrm{Bernoulli}(0.5), \label{additional_p_t}\\
    &Y_{i}=\theta(\bm{X}_{i})T_{i} + f(\bm{X}_i) + \varepsilon_{i} \quad (\varepsilon_i \sim \mathcal{N}(0, 1)),\label{additional_p_y}\\
    &f(\bm{X}_{i})=2\log(1+\mathrm{e}^{X_{i1}+X_{i2}+X_{i3}}),\\
    &\theta(\bm{X}_{i})=\sin(X_{i1})
\end{align}
where $\bm{X}_i \in \mathbb{R}^d$ ($d = 6$).
Note that the value of the function $\theta$ depends solely on $X_{i1}$.
\\
\textbf{Settings}: 
Suppose that, as prior knowledge, we know that $X_{1}$ is an important treatment effect modifier, while $X_{2},\ldots, X_{6}$ are not significant treatment effect modifiers. Given this, we consider the following covariance functions and hyperparameter values for the priors $\theta(\cdot) \sim \mathcal{GP}(0, C(\cdot, \cdot; \bm{\omega}_{\theta}))$ and
$f(\cdot)\sim \mathcal{GP}(0, C(\cdot, \cdot; \bm{\omega}_{f}))$.
\begin{align}
    &C(\bm{x}, \bm{x}'; \bm{\omega}_{\theta})=\exp\left\{-\sum_{k=1}^{6}\omega_{\theta,k}(x_{k}-x'_{k})^{2}\right\},\label{anisotropic}\\
    &C(\bm{x}, \bm{x}'; \bm{\omega}_{f})=\exp\left\{-\omega_{f}\|\bm{x}-\bm{x}'\|^{2}\right\}.\label{isotropic}\\
    &\bm{\omega}_{\theta}=[\omega_{\theta,1}, 10^{-10}, 10^{-10}, 10^{-10}, 10^{-10}, 10^{-10}]^{\top},\label{omega_theta}
\end{align}
The hyperparameter values in (\ref{omega_theta}) imply that 
the value of feature $X_{1}$ greatly contributes to the value of $\theta$, 
whereas the values of the other features, $X_{2},\ldots,X_{6}$, 
do not have a large impact on the value of $\theta$. \\
\textbf{Methods}: 
We considered the three Gaussian process-based Bayesian estimators:
\begin{itemize}
    \item Proposed (anisotropic): the proposed method that is based on the above prior knowledge about important features; it uses the covariance functions in (\ref{anisotropic}) and (\ref{isotropic}).
    We determined the values of $\omega_{\theta,1},\omega_{f}$ and $s_{\varepsilon}$ using Algorithm 1 and Algorithm 2.
    \item Proposed (isotropic): the proposed method that does not use any prior knowledge; as with the results in Section 6, it employs the covariance functions in the form of (\ref{isotropic}).
    The hyperparameters were determined using Algorithm 1 and Algorithm 2.
    \item Alaa+: the existing Gaussian-process-based model, which places a Gaussian process prior for each potential outcome \citep{alaa2018bayesian}\footnote{\url{https://github.com/vanderschaarlab/mlforhealthlabpub/tree/main/alg/causal_multitask_gaussian_processes_ite}}
\end{itemize}

As with the synthetic data experiments for binary treatment setup (Section 6.1), 
we constructed CATE estimator \(\hat{\theta}(\cdot)\), 
using training data with the size $n=100$. 
We then evaluated the performance by the MSE 
$\frac{1}{m}\sum_{i=1}^{m}(\hat{\theta}(\tilde{\bm{x}}_{i})-\theta(\tilde{\bm{x}}_{i}))^{2}$
for test data \(\tilde{\bm{x}}_{1},\ldots,\tilde{\bm{x}}_{m} \stackrel{i.i.d.}{\sim} p(\bm{x})\) with size $m=100$.\\
\textbf{Results}:
Table \ref{tab:result_addition} shows the results.
Compared to the proposed methods, the existing method, Alaa+, suffered from large estimation errors.
This result illustrates the weakness of Alaa+:
since it places a Gaussian process on each potential outcome,
it cannot deal with the cases
where the CATE is represented as a much simpler function than potential outcomes; 
in our setup, it is formulated with only a small number of features.
Compared with Alaa+,
the Proposed (anisotropic) achieved better estimation performance.
In particular, 
the MSE of the Proposed (anisotropic) was lower than the Proposed (isotropic),
demonstrating the effectiveness of utilizing the prior knowledge about treatment effect heterogeneity.

\begin{table}[t]
    \centering
    \caption{Simulation results on synthetic data generated according to (\ref{additional_p_x})-(\ref{additional_p_y}). The average and standard deviation of MSEs.}
    \begin{tabular}{c|c}
     Method & MSE \\
    \hline 
    \hline
     Proposed (anisotropic)     & $0.185\pm 0.188$ \\
     Proposed (isotropic) &  $0.240\pm 0.0745$\\
     Alaa+ & $0.244\pm 0.112$\\
     \hline
    \end{tabular}
    \label{tab:result_addition}
\end{table}

\begin{figure*}[t]
    \centering
    \includegraphics[keepaspectratio=true, width=\linewidth]{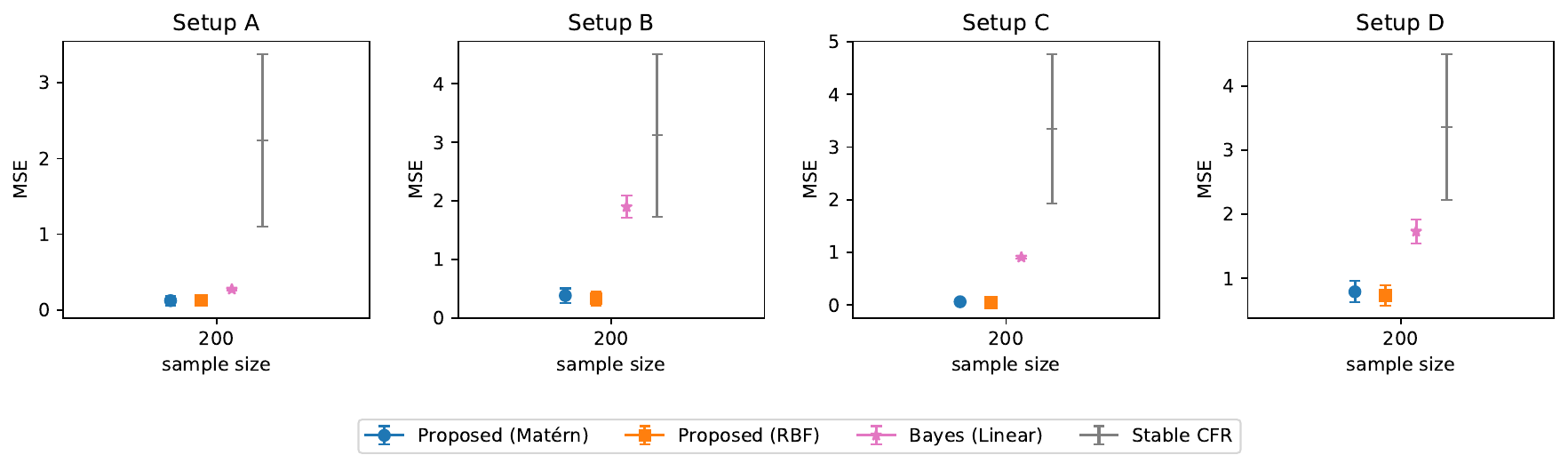}
    \caption{MSEs of the proposed methods and linear model Bayesian estimator and stableCFR (binary treatment setups and $n=200$). Lower is better.}
    \label{fig:AAAI2024_rebuttal_experiment}
\end{figure*}

\section{Additional Performance Comparison}
\label{sec:rebuttal_experiment}
To further investigate the empirical performance of the proposed method, we compared it with the Bayesian linear regression model and the stableCFR \cite{wu2023stable}, which is a recently proposed neural-network-based method.

Figure \ref{fig:AAAI2024_rebuttal_experiment} shows the results on the synthetic datasets used in Section 6.
Our method achieves higher accuracy than the Bayesian linear regression model because it employs nonlinear models to represent functions $\theta$ and $f$.
The stableCFR worked poorly in a small sample size setting due to its model complexity.

\end{document}